\documentclass[12pt]{article}
\usepackage{amsmath,amsfonts,amsthm}
\newtheorem{lemma}{Lemma}
\newtheorem{prop}[lemma]{Proposition}
\newtheorem{coro}[lemma]{Corollary}
\newtheorem{thm}[lemma]{Theorem}
\theoremstyle{remark}
\newtheorem{rem}[lemma]{Remark}

\def\Ha{\mathcal H}
\def\Ka{\mathcal K}
\def\Fe{\mathcal F}

\def\Tr{\mathrm{Tr}\,}
\def\states{\mathfrak S}
\def\<{\langle}
\def\>{\rangle}
\def\supp{\mathrm{supp}}
\begin{document}
\title{Preservation of a quantum  R\'enyi relative entropy implies existence of a recovery map}
\author{Anna Jen\v cov\'a\thanks{jenca@mat.savba.sk}\\ \small \emph{Mathematical Institute, Slovak Academy of Sciences}\\
\small \emph{ \v Stef\'anikova 49, 814 73 Bratislava, Slovakia}}
\date{}
\maketitle
\abstract{It is known that a necessary and sufficient condition for equality in the data processing inequality (DPI) for the  quantum relative entropy is the existence of a recovery map.  We show that equality in DPI for a sandwiched R\'enyi relative $\alpha$-entropy with $\alpha>1$ is also equivalent to this property. For the proof, we use an interpolating family of $L_p$-norms with respect to a state.}

\section{Introduction}

One of the key notions of quantum information theory is that of a relative entropy, or divergence, which can be seen as a measure of information-theoretic difference between quantum states. A fundamental property of a relative entropy $D$ is the  data processing inequality (DPI)
\begin{equation}\label{eq:dpi}
D(\Phi(\rho)\|\Phi(\sigma))\le D(\rho\|\sigma),
\end{equation}
which holds for any pair of quantum  states $\{\rho, \sigma\}$ and any quantum channel $\Phi$. This property expresses the fact that  a physical transformation, 
represented by the channel $\Phi$,  cannot increase distinguishability of states. It is clear that if there is a quantum channel $\Psi$  such that
\[
\Psi\circ\Phi(\sigma)=\sigma\mbox{ and }\Psi\circ\Phi(\rho)=\rho,
\]
then  equality in DPI holds.  Such a channel $\Psi$, if it exists, is called a recovery map for the triple $(\Phi,\rho,\sigma)$. The aim of the present work is to prove the opposite implication, in the case that $D$ is a quantum version of the R\'enyi relative entropy.  

One of the most important quantum relative entropies is the Umegaki relative entropy \cite{umegaki1962conditional}, which has an operational significance as a distinguishability measure on quantum states \cite{hipe1991properformula, ogna2000strong} and is related to  many other entropic quantities.
For two states $\rho$ and $\sigma$, it is defined as
\[
D_1(\rho\|\sigma)=\Tr \rho\left(\log(\rho)-\log(\sigma)\right),
\]
if the support $\supp(\rho)$ of $\rho$ is contained in $\supp(\sigma)$ and is infinite 
otherwise. Data processing inequality for $D_1$ was proved in \cite{lindblad1975completely, uhlmann1977relative, petz1984quasi}, recently also for positive trace-preserving maps \cite{mhre2015monotonicity}.  See also  \cite{ruskai2002inequalities} for a discussion including some related inequalities and equality conditions.

For $\alpha\in (0,1)\cup (1,\infty)$, the standard  quantum version of the  R\'enyi relative $\alpha$-entropy is defined by \cite{ohpe1993entropy} 
\[
D_\alpha(\rho\|\sigma)= \left\{\begin{array}{cc}    \frac1{\alpha-1} \log\left( \Tr \rho^\alpha\sigma^{1-\alpha}\right)& \mathrm{if }\ \alpha\in (0,1) \mbox{ or }\supp(\rho)\subseteq \supp(\sigma)\\ & \\
 \infty & \mathrm{otherwise}.
 \end{array}
 \right. 
\]
This quantity can be derived from  quantum $f$-divergences, or quasi-entropies, defined in \cite{petz1984quasi} and is a straightforward generalization of the classical family of R\'enyi relative  $\alpha$-entropies  \cite{renyi1961onmeasures}. The values for 
$\alpha=0,1,\infty$ can be obtained by taking limits, in particular, $\lim_{\alpha\to 1}D_\alpha=D_1$, the Umegaki relative entropy. The R\'enyi relative $\alpha$-entropy  satisfies DPI for $\alpha\in [0,2]$, see \cite{petz1984quasi,hmpb2011divergences}.
 For $\alpha\in (0,1)$, the relative $\alpha$-entropies appear in error exponents \cite{hayashi2007exponet, nagaoka2006theconverse, ansv2008chernoff,  hmo2008exponents} and as cutoff rates  \cite{himo2011onthequantum} 
in quantum hypothesis testing.

It was an important observation by Petz  that if $\Phi$ is a quantum channel and $\rho$, $\sigma$ are states such that $D_1(\rho\|\sigma)$ is finite, then the equality
\[
D_1(\Phi(\rho)\|\Phi(\sigma))= D_1(\rho\|\sigma)
\] 
is equivalent to existence of a recovery map for $(\Phi,\rho,\sigma)$.
 This was first proved in \cite{petz1986sufficient} in the general framework of von Neumann algebras, in the case that $\Phi$ is a restriction to a subalgebra. In \cite{petz1988sufficiency}, this result was extended to subsets of states and any channel $\Phi$. Here the Umegaki relative entropy was replaced by transition probability, which is closely related to the R\'enyi relative entropy $D_{1/2}$, but the proof, using integral representations of operator convex functions, can be easily 
 extended to $D_\alpha$ for all $\alpha\in (0,2)$, see \cite{jepe2006sufficiency,jepe2006idaqp, hmpb2011divergences}.  

In analogy with the classical notion of a sufficient statistic \cite{fisher1922onthemathematical}, in particular its characterization  as a statistical isomorphism (see e.g. \cite{strasser1985statistics}), Petz in \cite{petz1988sufficiency} called  a channel $\Phi$ sufficient with respect to  $\{\rho,\sigma\}$ if there is  a recovery map for $(\Phi,\rho,\sigma)$.
The fundamental result of \cite{petz1986sufficient} is thus a quantum generalization of the well-known characterization of classical sufficient statistics in terms of the Kullback-Leibler divergence, \cite{kule1951oninformation}. 

Sufficiency of channels, sometimes also called reversibility,  and its equivalent characterizations were subsequently studied by several authors. On one hand, the characterization by equality in DPI was extended to a large class of quantum divergences, such as $f$-divergences, Fisher information  and some distinguishability measures related to quantum hypothesis testing,  \cite{hmpb2011divergences,jencova2012reversibility}. On the other hand, it was proved that sufficiency is equivalent to a certain factorization structure of both the channel and the involved states, \cite{mope2004sufficient,  jepe2006sufficiency,jepe2006idaqp}.   See \cite{hmpb2011divergences} for a review of quantum $f$-divergences and conditions for sufficiency, and \cite{jencova2012reversibility, himo2016different} for some extensions and further results.       
Moreover, sufficiency has been studied also in the framework of quantum extension of Blackwell's comparison of statistical experiments \cite{buscemi2012comparison} and  in the setting of bosonic channels  \cite{shirokov2013reversibility}. 

The theory of sufficient channels proved useful  for finding  equality conditions in some  quantum theoretic inequalities, most notably for characterization of the quantum Markov property by equality in strong subadditivity of von Neumann  entropy, \cite{hjpw2004ssa,jepe2006sufficiency}. Further, equality conditions were obtained for Holevo quantity and related entropic inequalities, \cite{shirokov2013reversibility,shirokov2014onequalities} and for  convexity of some $f$-divergences, as well as in some related Minkowski inequalities, \cite{jeru2010unified}.  See also \cite{osiy2005quantum} for applications to secret sharing schemes.

Recently, another quantum version of  R\'enyi relative entropy was introduced in \cite{wwy2014strong} and independently in \cite{mldsft20130nquantum}. It is the so-called sandwiched R\'enyi relative $\alpha$-entropy, defined for $\alpha>0$, $\alpha\ne 1$ as 
\[
\tilde D_\alpha(\rho \| \sigma)=\left\{\begin{array}{cc} \frac{1}{\alpha-1}\log \Tr \left[\left( \sigma^{\frac{1-\alpha}{2\alpha}}\rho
\sigma^{\frac{1-\alpha}{2\alpha}}\right)^\alpha\right] & \mathrm{if }\ \supp(\rho)\subseteq \supp(\sigma)\\ & \\
 \infty & \mathrm{otherwise}
 \end{array}
 \right.
\]
Again, the values for $\alpha=0,1,\infty$ are obtained by taking limits, \cite{dale2014alimit, mldsft20130nquantum, wwy2014strong} and we have 
$\tilde D_1=D_1$ and $\tilde D_\infty=D_{max}$ , where 
\begin{equation}\label{eq:dmax}
D_{max}(\rho\|\sigma)=\log\inf\{\lambda>0, \rho\le \lambda\sigma\}
\end{equation}
 is the max-relative entropy, introduced in \cite{datta2009min}. Moreover, $\tilde D_\alpha$ satisfies DPI for $\alpha\in [1/2,\infty]$, \cite{mldsft20130nquantum, wwy2014strong, beigi2013sandwiched, frli2013monotonicity}. For $\alpha>1$ the sandwiched R\'enyi relative entropies $\tilde D_\alpha$ have an operational interpretation  as strong converse exponents in quantum hypothesis testing \cite{moog2015quantum}, quantum channel discrimination \cite{cmw2016strong} and classical-quantum channel coding \cite{moog2014strong}.

It is a natural question, mentioned also in  \cite{moog2015quantum} and  \cite{dawi2015quantum},   whether equality in DPI  for $\tilde D_\alpha$ characterizes sufficiency. An algebraic equality condition for $\alpha\in (1/2,1)\cup(1,\infty)$ was proved in the recent paper \cite{lrd2016data} and applications to R\'enyi  versions of Araki-Lieb inequalities and entanglement measures were found, but the question whether this implies existence of a recovery map remained open. Partial results for a larger family of $\alpha$-$z$-R\'enyi entropies 
 were also obtained in \cite{himo2016different}, under some restrictions on the triple
 $(\Phi,\rho,\sigma)$. 
 
In the present work, we answer this question in the affirmative for $\alpha>1$ and any 
 triple $(\Phi,\rho,\sigma)$ where $\Phi$ is a completely positive (or 2-positive) trace preserving map and $\supp(\rho)\subseteq \supp(\sigma)$. Since the sandwiched R\'enyi relative  entropies do not belong to the family of $f$-divergences, the techniques used e.g. in \cite{hmpb2011divergences} cannot be applied to this case. 
We will use  a method based on an interpolating family of $L_p$-norms, close to the approach introduced  in \cite{beigi2013sandwiched}. The complex interpolation method can also be utilized to extend the definition of $\tilde D_\alpha$ to infinite dimensions, as will be done in a subsequent paper.

\section{Preliminaries}

Throughout the paper, we will restrict to finite dimensional Hilbert spaces. For a  Hilbert space $\Ha$, we  denote the algebra of bounded linear  operators on $\Ha$ by $B(\Ha)$ and the cone of positive operators by $B(\Ha)^+$. States on $B(\Ha)$  are identified with  positive operators with unit trace, the set of all states will be denoted by $\states(\Ha)$.  For  $X\in B(\Ha)^+$, $\supp(X)$ denotes the support of $X$.

 For $p\ge 1$, the Schatten $p$-norm in $B(\Ha)$ is introduced as
\[
\|X\|_p=\left(\Tr |X|^p\right)^{1/p},\qquad X\in B(\Ha)
\]
and  we put $\|\cdot\|_\infty=\|\cdot\|$, the operator norm. The space $B(\Ha)$ equipped  with the norm $\|\cdot\|_2$ becomes a Hilbert space, with respect to the  Hilbert-Schmidt inner product
\[
\<X,Y\>=\Tr X^*Y,\qquad X,Y\in B(\Ha).
\]

If $\Ka$ is another Hilbert space and $\Phi:B(\Ha)\to B(\Ka)$ is a linear map, its adjoint with respect to the Hilbert-Schmidt inner product will be denoted by $\Phi^*$.  
The map $\Phi$ is trace-preserving if $\Tr \Phi(X)=\Tr X$ for all $X\in B(\Ha)$ and 
unital if $\Phi(I)=I$. Further, $\Phi$ is positive if $\Phi(B(\Ha)^+)\subseteq B(\Ka)^+$and  $n$-positive if $id_n\otimes \Phi$ is positive, where $id_n$ is the identity map on $B(\mathbb C^n)$. $\Phi$ is completely positive if it is $n$-positive for all $n\in \mathbb N$. A completely positive trace-preserving map is called a quantum channel. 
   
We will also need the notion of a conditional expectation  on  $B(\Ha)$, which is defined as a  positive unital map $E$ onto a subalgebra $\Fe\subseteq B(\Ha)$, with the property 
\begin{equation}\label{eq:condexp}
E(AXB)=AE(X)B, \qquad \forall A,B\in \Fe, X\in B(\Ha),
\end{equation}
see e.g. \cite{ohpe1993entropy,petz2008quantum}. Note that any conditional expectation is completely positive and  $E^2=E$.

\subsection{Non-commutative $L_p$-spaces with respect to a state}

The use of complex interpolation method appears in the study of the sandwiched  R\'enyi relative entropies  in the work by Beigi, \cite{beigi2013sandwiched}.
For our results, we use an interpolating  family of $L_p$-norms with respect to a state  $\sigma\in \states(\Ha)$, which are closely related but not the same as those introduced in \cite{beigi2013sandwiched}, see Remark \ref{rem:beigi} below. In a more general setting of von Neumann algebras, the corresponding $L_p$-spaces were introduced in \cite{trunov1979anoncommutative, zolotarev1982lpspaces} and are contained in a larger family of interpolation $L_p$-spaces defined in \cite{kosaki1984applications}.  

For $1\le p\le \infty$ and $1/p+1/q=1$, let $L_{p,\sigma}(\Ha)$ denote the linear space  of operators $Y\in B(\Ha)$, such that  $Y=\sigma^{1/2q}X\sigma^{1/2q}$ for some $X\in B(\Ha)$, equipped with the norm 
\[
\|Y\|_{p,\sigma}=\|X\|_p.
\]
It is clear that for any $p$,  $L_{p,\sigma}(\Ha)$ coincides with the subspace $L_\sigma(\Ha)$ of operators whose support and range lie in $\supp(\sigma)$, which can be identified with $B(\supp(\sigma))$. We will use this identification without further notice below, so that for example we will write  $\sigma^z$ for the operator $(\sigma|_{\supp(\sigma)})^z\oplus 0$ for any $z\in \mathbb C$. In this way, we may write
\begin{equation}\label{eq:pnorm}
\|Y\|_{p,\sigma}= \|\sigma^{\frac{1-p}{2p}}Y\sigma^{\frac{1-p}{2p}}\|_p,\qquad Y\in L_{\sigma}(\Ha).
\end{equation}
It follows from the properties of Schatten $p$-norms that this indeed defines a norm.
For $p=1$, the norm $\|\cdot\|_{1,\sigma}$ is just  $\|\cdot\|_1$ restricted to $L_{\sigma}(\Ha)$ and 
 for $p=\infty$,
\begin{equation}\label{eq:linfty}
\|Y\|_{\infty,\sigma}=\|\sigma^{-1/2}Y\sigma^{-1/2}\|.
\end{equation}
 Note also that if $\rho\in\states(\Ha)$, then $\rho\in L_{\sigma}(\Ha)$ if and only if $\rho\le \lambda\sigma$ for some $\lambda>0$ and in this case
\[
\|\rho\|_{\infty,\sigma}=\inf\{\lambda>0, \rho\le \lambda\sigma\}=2^{D_{max}(\rho\|\sigma)}.
\]

For $Z,Y\in L_{\sigma}(\Ha)$, we define
\[
\<Z,Y\>_\sigma=\Tr Z\sigma^{-1/2}Y\sigma^{-1/2}.
\]
Let $1\le p\le \infty$ and $1/p+1/q=1$, then
\[
\<Z,Y\>_\sigma= \Tr \left(\sigma^{\frac{1-q}{2q}}Z\sigma^{\frac{1-q}{2q}}\right)\left(\sigma^{\frac{1-p}{2p}}Y\sigma^{\frac{1-p}{2p}}\right).
\]
 By duality of the corresponding Schatten norms, the following duality holds:
\begin{equation}\label{eq:duality}
\|Y\|_{p,\sigma}=\sup_{\|Z\|_{q,\sigma}\le 1} |\<Z,Y\>_\sigma|.
\end{equation}
Using Remark \ref{rem:beigi} below, this can be also obtained from \cite{beigi2013sandwiched}.
 The space $L_{2,\sigma}(\Ha)$ is a Hilbert space with respect to the inner product
\[
(Y_1,Y_2)\mapsto \<Y_1^*,Y_2\>_\sigma.
\]

\begin{rem}\label{rem:beigi}  Let us denote the norm in \cite{beigi2013sandwiched} by $\|\cdot\|_{p,\sigma,B}$, then it is easy to see that for $Y\in L_{\sigma}(\Ha)$, 
\[
\|Y\|_{p,\sigma}=\|\sigma^{-1/2}Y\sigma^{-1/2}\|_{p,\sigma,B}
\]
and one can switch the results in \cite{beigi2013sandwiched} to our setting by means of the map $\Gamma_\sigma: X\mapsto \sigma^{1/2}X\sigma^{1/2}$ and its inverse.  
We often use this relation to our advantage, but our definition  seems preferable for the present work, since it leads to simpler expressions. For example, the sandwiched R\'enyi relative entropy is obtained directly as logarithm of the norm and the Petz recovery map is  the adjoint with respect to the inner product $\<\cdot,\cdot\>_\sigma$ (see (\ref{eq:adjointsigma}) and (\ref{eq:petzrecovery}) below). 

On the other hand, it was proved in \cite{zolotarev1982lpspaces, kosaki1984applications} that the spaces $L_{p,\sigma}(\Ha)$ can be obtained by the complex interpolation method and most of the results of this and the next subsection can be obtained from these works. Nevertheless, in order not to go into unnecessary technicalities, we prefer to give direct proofs or use the relation to \cite{beigi2013sandwiched}. Let us also remark that this family of norms was studied also in \cite{olze1999hypercontractivity}, where further results can be found. 

\end{rem}

We next define certain operator valued analytic functions on the stripe $S=\{z\in \mathbb C,\ 0\le Re(z)\le 1\}$
 which will be useful for evaluation of the norm $\|\cdot\|_{p,\sigma}$. So let $Y\in L_{\sigma}(\Ha)$ and $1<p<\infty$. Let $X:=\sigma^{-1/2q}Y\sigma^{-1/2q}$ and let $X=U|X|$ be the polar decomposition. We will denote 
\begin{equation}\label{eq:fyp}
f_{Y,p}(z)=\|Y\|_{p,\sigma}^{1-zp}\sigma^{(1-z)/2}U|X|^{zp}\sigma^{(1-z)/2},\qquad z\in S.
\end{equation}
Then $f_{Y,p}$ is a bounded continuous function $S\to L_\sigma(\Ha)$, holomorphic in the interior and $f_{Y,p}(1/p)=Y$. Moreover, we have
\[
f_{Y,p}(it)=\|Y\|_{p,\sigma}^{1-itp}\sigma^{1/2}\left(\sigma^{-it/2}U|X|^{itp}\sigma^{-it/2}\right)\sigma^{1/2}
\]
and by \eqref{eq:linfty},
\begin{equation}\label{eq:attain1}
\|f_{Y,p}(it)\|_{\infty,\sigma}= \|Y\|_{p,\sigma}\|\sigma^{-it/2}U|X|^{itp}\sigma^{-it/2}\|=\|Y\|_{p,\sigma},\quad t\in \mathbb R.
\end{equation}
Similarly, for all  $t\in \mathbb R$,
\begin{equation}\label{eq:attain2}
\|f_{Y,p}(1+it)\|_1=\|Y\|_{p,\sigma}^{1-p}\|\sigma^{-it/2}U|X|^{p+pit}\sigma^{-it/2}\|_1=\|Y\|_{p,\sigma},\end{equation}
the last equality is obtained from the fact that $\|\cdot\|_1$ is unitarily invariant. The following lemma will be needed in the sequel. Note also that the inequality part follows directly from the version of Hadamard three lines theorem proved in \cite{beigi2013sandwiched}.

\begin{lemma}\label{lemma:inf} Let $h: S\to L_{\sigma}(\Ha)$ be a bounded continuous function, holomorphic in the interior of $S$. Then for any $\theta\in [0,1]$, we have
\[
\|h(\theta)\|_{1/\theta,\sigma}\le \max\{\sup_t \|h(it)\|_{\infty,\sigma}, \sup_t\|h(1+it)\|_1\}.
\]
Moreover, if equality is attained at some $\theta\in (0,1)$, then it holds for all $\theta\in [0,1]$.

\end{lemma}

\begin{proof}
Let $\theta\in (0,1)$. By the duality \eqref{eq:duality}, there is some  $Z$ in the unit ball of $L_{1/(1-\theta),\sigma}(\Ha)$ such that 
$\|h(\theta)\|_{1/\theta,\sigma}=\<Z,h(\theta)\>_\sigma$. Consider the corresponding function $f_{Z,1/(1-\theta)}$, so that $f_{Z,1/(1-\theta)}(1-\theta)=Z$ and by  (\ref{eq:attain1}) and (\ref{eq:attain2}), we have for all $t\in \mathbb R$, 
\[
 \|f_{Z,1/(1-\theta)}(it)\|_{\infty,\sigma}=\|f_{Z,1/(1-\theta)}(1+it)\|_1=\|Z\|_{1/(1-\theta),\sigma}=1.
\]
Put $H(z):=\<f_{Z,1/(1-\theta)}(1-z),h(z)\>_\sigma$, $z\in S$, then  $H:S\to \mathbb C$ is bounded, continuous and holomorphic in the interior of $S$. 
We have for $t\in \mathbb R$,
\[
|H(it)|\le \|f_{Z,1/(1-\theta)}(1-it)\|_1\|h(it)\|_{\infty,\sigma}=\|h(it)\|_{\infty,\sigma}
\]
and 
\[
|H(1+it)|\le \|f_{Z,1/(1-\theta)}(-it)\|_{\infty,\sigma}\|h(1+it)\|_1=\|h(1+it)\|_1.
\]
Let $M:=\max\{\sup_t \|h(it)\|_{\infty,\sigma}, \sup_t\|h(1+it)\|_1\}$, then it follows by the Hadamard three lines theorem that $|H(z)|\le M$ for all $z\in S$. 
The inequality part of the statement now follows by $\|h(\theta)\|_{1/\theta,\sigma}=H(\theta)$.

If equality is attained at $\theta$, then  by the maximum modulus principle,  $H$ must be a constant, so that $H(z)=M$ for all $z\in S$.
For any $\lambda\in [0,1]$ we then have  
\begin{align*}
M&=H(\lambda)\le \|f_{Z,1/(1-\theta)}(1-\lambda)\|_{1/(1-\lambda),\sigma}\|h(\lambda)\|_{1/\lambda,\sigma}\\
&\le \|h(\lambda)\|_{1/\lambda,\sigma}\le M,
\end{align*}
where the last two inequalities follow by the first part of the proof.

\end{proof}

\subsection{Positive trace-preserving maps and their duals}

Let $\Ka$ be a finite dimensional Hilbert space and let $\Phi: B(\Ha)\to B(\Ka)$ be a positive trace-preserving map, then its adjoint $\Phi^*$  is a unital positive map. Let $\sigma\in \states(\Ha)$, let $P\in B(\Ha)$ be 
the projection onto $\supp(\sigma)$ and $Q\in B(\Ka)$ the projection onto $\supp(\Phi(\sigma))$. Then 
\[
0=\Tr \Phi(\sigma)Q^\perp=\Tr \sigma \Phi^*(Q^\perp),
\]
 where $Q^\perp=I-Q$, so that  $\Phi^*(Q^\perp)\le P^\perp$. Let  $X$ be a positive 
element in $L_\sigma(\Ha)$, then
\[
\Tr \Phi(X)Q^\perp=\Tr X\Phi^*(Q^\perp)=0,
\]
so that $\supp(\Phi(X))\subseteq \supp(\Phi(\sigma))$. Since $L_\sigma(\Ha)$ is generated by positive elements, it follows that $\Phi$ maps $L_{\sigma}(\Ha)$ into $L_{\Phi(\sigma)}(\Ka)$. The following result is easily obtained by Remark \ref{rem:beigi} and the proof of Theorem 6 in \cite{beigi2013sandwiched}. Note that as observed in \cite{mhre2015monotonicity}, only  positivity of $\Phi$ is used in this proof.    

\begin{prop}\label{prop:contraction} $\Phi$ defines a contraction $L_{p,\sigma}(\Ha)\to L_{p,\Phi(\sigma)}(\Ka)$  for  all $1\le p\le \infty$.

\end{prop}

Let $\Phi_\sigma$ denote the adjoint of $\Phi$ with respect to $\<\cdot,\cdot\>_\sigma$, that is, $\Phi_\sigma: L_{\Phi(\sigma)}(\Ka)\to L_\sigma(\Ha)$ satisfies 
\begin{equation}\label{eq:adjointsigma}
\<\Phi_\sigma(X),Y\>_{\sigma}=\<X,\Phi(Y)\>_{\Phi(\sigma)},\qquad X\in L_{\Phi(\sigma)}(\Ka), Y\in L_\sigma(\Ha).  
\end{equation}
It is easy to see that
\begin{equation}\label{eq:petzrecovery}
\Phi_\sigma(X)= \sigma^{1/2}\Phi^*(\Phi(\sigma)^{-1/2}X\Phi(\sigma)^{-1/2})\sigma^{1/2}.
\end{equation}
Note that this is precisely the (adjoint of) the dual map of $\Phi^*$ as defined by Petz \cite{petz1988sufficiency}. Note also that $\Phi_\sigma$ is trace-preserving, positive and $\Phi_\sigma(\Phi(\sigma))=\sigma$, so that Proposition \ref{prop:contraction} implies that $\Phi_\sigma$ is  a contraction $L_{p,\Phi(\sigma)}(\Ka)\to L_{p,\sigma}(\Ha)$ for any $p$. Moreover, for any $n\in \mathbb N$,  
$\Phi_\sigma$ is $n$-positive if and only if $\Phi$ is. 

\subsection{Quantum channels and sufficiency}

Let $\Phi:B(\Ha)\to B(\Ka)$ be a quantum channel and let $\sigma\in \states(\Ha)$.
The next result shows that the channel  $\Phi_\sigma$ defines a universal recovery map 
 for $\Phi$ and $\sigma$. For this reason, $\Phi_\sigma$ is sometimes called the Petz recovery map.

\begin{thm}\label{thm:sufficient}\cite{petz1988sufficiency} Let $\Phi:B(\Ha)\to B(\Ka)$ be a quantum channel and let $\rho,\sigma\in \states(\Ha)$ be such that $\supp(\rho)\subseteq \supp(\sigma)$. Then  $\Phi$ is sufficient with respect to $\{\rho,\sigma\}$ if and only if  $\Phi_\sigma\circ \Phi(\rho)=\rho$.
\end{thm}

By restriction to $\supp(\sigma)$, we may assume that $\sigma$ is invertible. 
In this case, $\Omega:=\Phi_\sigma\circ\Phi$ is a channel on $B(\Ha)$ and  all states 
$\rho\in \states(\Ha)$ such that $\Phi$ is sufficient with respect to $\{\rho,\sigma\}$ are precisely all  invariant states  of  $\Omega$. Note that  by \eqref{eq:petzrecovery}, $\sigma$ is an invertible 
 invariant state of $\Omega$. The following results are known.

\begin{thm}\label{thm:factorization} Let $\Omega:B(\Ha)\to B(\Ha)$ be a channel such that $\Omega(\sigma)=\sigma$ for some invertible $\sigma\in \states(\Ha)$. Then there are Hilbert spaces $\Ha_n^L$, $\Ha_n^R$ and  a unitary operator $U: \Ha \to \bigoplus_n \Ha_n^L\otimes \Ha_n^R$,  such that
\begin{enumerate} 
\item[(i)] the set $\Fe_{\Omega^*}$ of fixed points of $\Omega^*$ has the form
\[
\Fe_{\Omega^*}=U^*\left(\bigoplus_n B(\Ha_n^L)\otimes I_{\Ha_n^R}\right) U.
\] 
\item[(ii)] the set $\Fe_\Omega$ of fixed points of $\Omega$ has the form
\[
\Fe_{\Omega}=U^*\left(\bigoplus_n B(\Ha_n^L)\otimes \sigma_n^R\right) U,
\]
for some invertible $\sigma_n^R\in \states(\Ha_n^R)$. 

\end{enumerate}
\end{thm}

\begin{proof} 
Part (i)  is quite standard. It is proved e.g. in  \cite[Example 9.4]{petz2008quantum}
that $\Fe_{\Omega^*}$ is a subalgebra and the sequence $\frac1n\sum_{k=0}^{n-1}(\Omega^*)^k$, where $(\Omega^*)^k$ denotes the $k$-fold composition $\Omega^*\circ\dots\circ\Omega^*$, converges to a conditional expectation
$E$ onto $\Fe_{\Omega^*}$.   Since any subalgebra in $B(\Ha)$ is isomorphic to a direct sum of full matrix algebras, it must have the stated form.

The proof of part (ii) can be found in \cite{wolf2012quantum}. We give another proof of (ii) using some arguments of \cite{hjpw2004ssa}, which is perhaps more straightforward.

Note that the conditional expectation $E$ satisfies $E^*(\sigma)=\sigma$ and $\Omega^*\circ E=E\circ\Omega^*=E$. 
Let now $Y \in \Fe_\Omega$. Then $\Omega^k(Y)=Y$ for any $k\in \mathbb N$, so that 
$E^*(Y)=Y$. Conversely, let $E^*(Y)=Y$, then 
\[
\Omega(Y)=\Omega(E^*(Y))=(E\circ \Omega^*)^*(Y)=E^*(Y)=Y. 
\]
We have proved that $\Fe_\Omega=\Fe_{E^*}$. 

Let $P_n:\Ha\to \Ha^L_n\otimes \Ha^R_n$ be orthogonal projections, then $Q_n:=U^*P_nU$ are central projections in $\Fe_{\Omega^*}=\Fe_E$, the range of $E$. Therefore, we must have for all $X\in B(\Ha)$,
\begin{align*}
E(X)&=\sum_{m,n} E(Q_mXQ_n)=\sum_{m,n} Q_mE(Q_mXQ_n)Q_n\\
&=\sum_{m,n} Q_nE(Q_mXQ_n)Q_n=\sum_n Q_nE(Q_nXQ_n)Q_n,
\end{align*}
where we used the property \eqref{eq:condexp} in the second and in the last equality. 
Further, by the same property, we have for any $X_n\in B(\Ha^L_n)$, $Y_n\in B(\Ha^R_n)$
\begin{align*}
E\left(U^*(X_n\otimes Y_n)U\right)&=U^*(X_n\otimes I_{\Ha_n^R})UE(U^*(I_{\Ha_n^L}\otimes Y_n)U)\\
&=E(U^*(I_{\Ha^L_n}\otimes Y_n)U)U^*(X_n\otimes I_{\Ha^R_n})U
\end{align*}
so that $E(U^*(I_{\Ha^L_n}\otimes Y_n)U)$ lies in the center of $U^*(B(\Ha^L_n)\otimes I_{\Ha^R_n})U$ and therefore is a multiple of $Q_n$. It follows that there is some linear functional $\psi_n$ on $B(\Ha_n^R)$ such that 
\[
E(U^*(I_{\Ha^L_n}\otimes Y_n)U)=\psi_n(Y_n)Q_n.
\]
Since $E$ is positive and unital, $\psi_n$ must be such as well and there must be some  $\sigma^R_n\in \states(\Ha^R_n)$, such that $\psi_n(Y_n)=\Tr [\sigma^R_nY_n]$, $Y_n\in B(\Ha_n^R)$. Let us denote 
$\phi_{\sigma_n^R}: Y_n\mapsto \Tr[\sigma_n^RY_n] I_{\Ha_n^R}$, $Y_n\in B(\Ha_n^R)$.  
It follows that
\[
E=U^*\left(\sum_n (id_{B(\Ha^L_n)}\otimes \phi_{\sigma^R_n})(P_nU\cdot U^*P_n)\right)U
\]
and
\[
E^*=U^*\left(\sum_n (id_{B(\Ha^L_n)}\otimes \phi^*_{\sigma^R_n})(P_nU\cdot U^*P_n)\right)U.
\]
It is now clear that $\Fe_\Omega=\Fe_{E^*}$ has the form as in (ii), with $\sigma_n^R$ as above. Since $\sigma\in \Fe_\Omega$ is invertible, $\sigma_n^R$ must be invertible as well, for all $n$.  

\end{proof}

As a consequence, we obtain a following characterization of sufficient channels. 
Similar results were already obtained in \cite{mope2004sufficient} and \cite{jepe2006sufficiency}. 

\begin{coro}\label{coro:factorization}
Let $\Phi : B(\Ha)\to B(\Ka)$ be a channel and let $\sigma\in \states(\Ha)$ be invertible. Then there is a unitary $U$ and factorizations
\[
U\Ha=\bigoplus_n \Ha_n^L\otimes \Ha_n^R,\qquad U\sigma U^*=\bigoplus_n A_n^L\otimes \sigma_n^R,
\]
where $A_n^L\in B(\Ha_n^L)^+$ and $\sigma_n^R\in \states(\Ha_n^R)$ are invertible, such that for any $\rho\in \states(\Ha)$, $\Phi$ is sufficient with respect to $\{\rho,\sigma\}$ if and only if
\[
U\rho U^*=\bigoplus_n B_n^L\otimes \sigma_n^R
\]
for some $B_n^L\in B(\Ha_n^L)^+$.

\end{coro}

\begin{proof} Follows by Theorem \ref{thm:sufficient} and Theorem \ref{thm:factorization} applied to the channel $\Omega=\Phi_\sigma\circ\Phi$.

\end{proof}

\section{The main result}

Let $\rho,\sigma\in \states(\Ha)$. It is easy to see that for $\alpha>1$ we may define the sandwiched R\'enyi relative  entropies by
\[
\tilde D_\alpha(\rho\|\sigma)=\left\{\begin{array}{cc}\frac{\alpha}{\alpha-1}\log \|\rho\|_{\alpha,\sigma}& \ \mbox{if } \rho\in L_{\alpha, \sigma}(\Ha)\\
 \ & \\
\infty & \mbox{otherwise}
\end{array}\right.
\]
Similarly as in \cite{beigi2013sandwiched,mhre2015monotonicity}, the data processing inequality for $\tilde D_\alpha$, $\alpha> 1$ and  for any positive trace-preserving map $\Phi$ is  an  obvious consequence of  Proposition \ref{prop:contraction}.  
Assume now that $\Phi$ is a channel. It is clear that if $\Phi$ is  sufficient with respect to $\{\rho,\sigma\}$, then equality 
in DPI must be attained. We will show that the converse is also true. 

\begin{thm}\label{thm:main} Let $\Phi:B(\Ha)\to B(\Ka)$ be a channel and let $1<\alpha<\infty$.  Let $\rho,\sigma\in \states(\Ha)$ be such that $\supp(\rho)\subseteq \supp(\sigma)$.   
 Then $\Phi$ is sufficient with respect to $\{\rho,\sigma\}$ if and only if 
\[
\tilde D_\alpha(\Phi(\rho)\|\Phi(\sigma))= \tilde D_\alpha(\rho\|\sigma).
\]

\end{thm}

The proof is based on the following lemmas. We assume below that $\rho,\sigma\in \states(\Ha)$ and $\supp(\rho)\subseteq \supp(\sigma)$.

\begin{lemma}\label{lemma:banff2} Let $\Phi: B(\Ha)\to B(\Ka)$ be a positive trace-preserving map. Then $\tilde D_2(\Phi(\rho)\|\Phi(\sigma))= \tilde D_2(\rho\|\sigma)$ if and only if $\Phi_\sigma\circ\Phi(\rho)=\rho$.

\end{lemma}

\begin{proof} Assume that the first equality  holds.
Then since $L_{2,\sigma}(\Ha)$ and $L_{2,\Phi(\sigma)}(\Ka)$ are Hilbert spaces and $\Phi_\sigma\circ \Phi$ is a contraction,
\[
\|\rho\|_{2,\sigma}^2=\|\Phi(\rho)\|_{2,\Phi(\sigma)}^2=\<\Phi(\rho),\Phi(\rho)\>_{\Phi(\sigma)}=\<\Phi_\sigma\circ \Phi(\rho),\rho\>_\sigma\le \|\rho\|_{2,\sigma}^2. 
\]
By the equality condition in Schwarz inequality, it follows that we must have $\Phi_\sigma\circ \Phi(\rho)=\rho$.

\end{proof}

\begin{lemma}\label{lemma:isometry} Let $\Phi:B(\Ha)\to B(\Ka)$ be a positive trace-preserving map and let $p>1$. Let $Y\in L_\sigma(\Ha)$ be such that $\|\Phi(Y)\|_{p,\Phi(\sigma)}=\|Y\|_{p,\sigma}$ and let $f_{Y,p}$ be as in \eqref{eq:fyp}.  
Then
\[
\|\Phi(f_{Y,p}(\theta))\|_{1/\theta,\Phi(\sigma)}=\|f_{Y,p}(\theta)\|_{1/\theta,\sigma},\qquad \forall\ \theta\in (0,1).
\]
\end{lemma}

\begin{proof} The function $\Phi\circ f_{Y,p}:S\to L_{\Phi(\sigma)}(\Ka)$ is bounded, continuous, holomorphic in the interior of $S$ and $\Phi\circ f_{Y,p}(1/p)=\Phi(Y)$. Using  Lemma \ref{lemma:inf}, Proposition \ref{prop:contraction} and the equalities \eqref{eq:attain1} and \eqref{eq:attain2}, we obtain
\begin{align*}
\|\Phi(Y)\|_{p,\Phi(\sigma)} &\le \max\{\sup_{t\in \mathbb R}\|\Phi(f_{Y,p}(it))\|_{\infty,\Phi(\sigma)},\sup_{t\in \mathbb R}\|\Phi(f_{Y,p}(1+it))\|_1\}\\
& \le \max\{\sup_{t\in \mathbb R}\| f_{Y,p}(it)\|_{\infty,\sigma},\sup_{t\in \mathbb R}\|f_{Y,p}(1+it)\|_1\} \\
&  =\|Y\|_{p,\sigma}=\|\Phi(Y)\|_{p,\Phi(\sigma)}.
\end{align*}
The statement now follows by the equality part in Lemma \ref{lemma:inf}.

\end{proof}

\noindent \emph{Proof of Theorem \ref{thm:main}}.  Put $p:=\alpha$, $1/p+1/q=1$. Assume that the equality holds, that is, $\|\Phi(\rho)\|_{p,\Phi(\sigma)}=\|\rho\|_{p,\sigma}$. By putting 
$\theta=1/2$ in Lemma \ref{lemma:isometry}, we obtain that 
\[
\|\Phi(f_{\rho,p}(1/2))\|_{2,\Phi(\sigma)}=\|f_{\rho,p}(1/2)\|_{2,\sigma}.
\] 
Let $X=\sigma^{-1/2q}\rho\sigma^{-1/2q}$, so that  $f_{\rho,p}(1/2)= \|\rho\|_{p,\sigma}^{1-p/2}\sigma^{1/4}X^{p/2}\sigma^{1/4}$, and  
put 
\[
\tau =\left(\Tr \sigma^{1/4}X^{p/2}\sigma^{1/4}\right)^{-1}\sigma^{1/4}X^{p/2}\sigma^{1/4}.
\]
Then $\tau\in \states(\Ha)$ and since $\tau$ is a constant multiple of $f_{\rho,p}(1/2)$, we have $\|\Phi(\tau)\|_{2,\Phi(\sigma)}=\|\tau\|_{2,\sigma}$. By Lemma \ref{lemma:banff2}, this implies that $\Phi$ is sufficient with respect to $\{\tau,\sigma\}$.
Let 
\[
U\Ha=\bigoplus_n \Ha_n^L\otimes \Ha_n^R,\qquad 
U\sigma U^*=\bigoplus_n A_n^L\otimes \sigma_n^R
\]
be a factorization as in Corollary \ref{coro:factorization}, then
$U\tau U^*=\bigoplus_n B_n^L\otimes \sigma_n^R$ for some $B_n^L\in B(\Ha_n^L)^+$. This entails that 
\[
UXU^* =\bigoplus_n C_n^L\otimes (\sigma_n^R)^{1/p}
\]
for some suitable $C_n^L\in B(\Ha_n^L)^+$, and consequently $\rho=\sigma^{1/2q}X\sigma^{1/2q}$ has a factorization of the form required in Corollary \ref{coro:factorization}. Hence $\Phi$ is sufficient with respect to $\{\rho,\sigma\}$.

\qed

\section{Concluding remarks}\label{sec:conclusion}

We have proved that equality in DPI  for $\tilde D_\alpha$ implies sufficiency for $\alpha>1$, but it is still not clear whether this is true for $\alpha\in (1/2,1)$, where DPI holds and the methods based on non-commutative $L_p$-norms can no longer be used. Note that the value $\alpha=1/2$ is excluded  since it is known that equality does not imply sufficiency in this case, \cite{moog2015quantum}. 

An algebraic equality condition for the range $\alpha\in (1/2,1)\cup (1,\infty)$  was obtained in the recent paper \cite{lrd2016data}. For $\alpha=2$, this algebraic condition is just as in our Lemma \ref{lemma:banff2}, but the proof in  \cite{lrd2016data} requires the map to be completely positive. We also remark that, as already observed in \cite{djw2015banff}, this algebraic condition can be obtained in the $L_p$-space framework for all $\alpha>1$ assuming only positivity of $\Phi$, using the fact that the spaces are uniformly convex and hence the norm of any element is attained at a unique point of the dual unit sphere. A similar uniqueness lies also at the core of the proof  in \cite{lrd2016data}.

As a final remark, note that for the main result, it is not necessary that $\Phi$ is completely positive. Indeed, the only place where more than positivity of $\Phi$ is required is the proof of Theorem \ref{thm:factorization}, but a closer look at the proof of \cite[Example 9.4]{petz2008quantum} shows that it would be enough to assume that the map $\Omega^*$ is a Schwarz map, that is, satisfying the inequality
\[
\Omega^*(X^*X)\ge \Omega^*(X^*)\Omega^*(X),\qquad \forall X\in B(\Ha).
\]
This holds if both $\Phi$ and $\Phi_\sigma$ are adjoints of a Schwarz map, which  is true for all $\sigma\in \states(\Ha)$ if and only if $\Phi$ is 2-positive, \cite[Proposition 2]{jencova2012reversibility}.


\section*{Acknowledgements} The research started from discussions with N. Datta and M. M. Wilde during the BIRS workshop 'Beyond IID in Information Theory', held in July 2015 in Banff, Canada. I am grateful to these colleagues for turning my attention to this problem and for useful discussions at the workshop and afterwards.  I am also indebted to the anonymous referees, whose valuable remarks and suggestions helped to improve the paper.
 The support from the grant No. VEGA 2/0069/16 is acknowledged as well.


\end{document}